\documentclass[10pt,draftcls,onecolumn]{IEEEtran}

\usepackage[active]{srcltx}
\usepackage{amsthm}
\usepackage{mathrsfs}
\usepackage{amsfonts}
\usepackage{color}
\usepackage{amsmath,amsthm}
\usepackage{amssymb}
\usepackage{amsfonts}
\usepackage[noadjust]{cite}
\usepackage{subfigure}
\usepackage{mathrsfs}
\usepackage{graphicx}
\usepackage{amsmath}
\usepackage{amsfonts}
\usepackage{latexsym,url}
\usepackage{amsmath,amssymb,array,graphicx,palatino,verbatim}
\usepackage{setspace}
{\makeatletter
 \gdef\xxxmark{%
   \expandafter\ifx\csname @mpargs\endcsname\relax 
     \expandafter\ifx\csname @captype\endcsname\relax 
       \marginpar{xxx}
     \else
       xxx 
     \fi
   \else
     xxx 
   \fi}
 \gdef\xxx{\@ifnextchar[\xxx@lab\xxx@nolab}
 \long\gdef\xxx@lab[#1]#2{{\bf [\xxxmark #2 ---{\sc #1}]}}
 \long\gdef\xxx@nolab#1{{\bf [\xxxmark #1]}}
}

\let\realbfseries=\bfseries
\def\bfseries{\realbfseries\boldmath}

\newif\ifabstract
\abstracttrue
\newif\iffull
\ifabstract \fullfalse \else \fulltrue \fi

\let\epsilon=\varepsilon

\newtheorem{theorem}{Theorem}
\newtheorem{definition}{Definition}
\newtheorem{lemma}{Lemma}

\newtheorem{fact}{Fact}

\newcommand{\X}{\mathcal{X}}
\newcommand{\Y}{\mathcal{Y}}
\newcommand{\Z}{\mathcal{Z}}

\renewcommand{\P}{\mathcal{P}}

\newcommand{\E}{\mathcal{E}}
\newcommand{\V}{\mathcal{V}}
\newcommand{\W}{\mathcal{W}}
\newcommand{\U}{\mathcal{U}}
\newcommand{\C}{\mathcal{C}}

\begin{document}
%

\title{A New Upper Bound on the Average Error Exponent for Multiple-Access Channels}

\author{
\authorblockN{Ali Nazari, S. Sandeep Pradhan and Achilleas Anastasopoulos\\}
\authorblockA{Electrical Engineering and Computer Science Dept. \\
University of Michigan, Ann Arbor, MI 48109-2122, USA\\
E-mail: \{anazari,pradhanv,anastas\}@umich.edu} 
\thanks{This work was supported by NSF ITR grant CCF-0427385.}}

\maketitle

\begin{abstract}
A new lower bound for the average probability or error for a
two-user discrete memoryless (DM) multiple-access channel (MAC) is
derived. This bound has a structure very similar to the well-known
sphere packing packing bound derived by Haroutunian. However, since
explicitly imposes independence of the users' input distributions
(conditioned on the time-sharing auxiliary variable) results in a
tighter sphere-packing exponent in comparison to Haroutunian's.
Also, the relationship between average and maximal error
probabilities is studied. Finally, by using a known sphere packing
bound on the maximal probability of error, a lower bound on the
average error probability is derived.
\end{abstract}


\section{Introduction} \label{intro}
One of the most important practical questions which arises when we
are designing or using an information transmission or processing
system is: How much information can this system transmit or process
in a given time? Information theory, developed by Claude E. Shannon
during World War II, defines the notion of channel capacity and
provides a mathematical model by which one can compute it.
Basically, Shannon coding theorem and all newer versions of it treat
the question of how much data can be reliably communicated from one
point, or sets of points, to another point or sets of points.

The class of channels to be considered include multiple transmitter
and a single receiver. The received signal is corrupted both by
noise and by mutual interference between the transmitters. Each of
transmitters is fed by an information source, and each information
source generates a sequence of messages. More specifically, a
two-user DM-MAC is defined by a stochastic matrix\footnote{We use
the following notation throughout this work. Script capitals
$\mathcal{U}$, $\mathcal{X}$, $\mathcal{Y}$, $\mathcal{Z}$,$\ldots$
denote finite, nonempty sets. To show the cardinality of a set
$\mathcal{X}$, we use $|\mathcal{X}|$. We also use the letters $P$,
$Q$,$\ldots$  for probability distributions on finite sets, and $U$,
$X$, $Y$,$\ldots$ for random variables.}
$W:\X \times \Y \rightarrow \Z$, where the input alphabets, $\X$,
$\Y$, and the output alphabet, $\Z$, are finite sets. The channel
transition probability for sequences of length $n$ is given by
\begin{align}
W^n\left(\mathbf{z}|\mathbf{x},\mathbf{y}\right) \triangleq
\prod_{i=1}^n W\left(z_i|x_i,y_i\right)
\end{align}
\indent where
\begin{align*}
\mathbf{x}\triangleq \left(x_1,...,x_n\right) \in \mathcal{X}^n,
\mathbf{y}\triangleq\left(y_1,...,y_n\right) \in \mathcal{Y}^n
\end{align*}
\indent and
\begin{align*}
\mathbf{z}\triangleq \left(z_1,...,z_n\right) \in \mathcal{Z}^n.
\end{align*}
It has been proven, by Ahlswede~\cite{Ahlswede71} and
Liao's~\cite{Liao} coding theorem, that for any
$\left(R_X,R_Y\right)$ in the interior of a certain set
$\mathcal{C}$, and for all sufficiently large $n$, there exists a
multiuser code with an arbitrary small average probability of error.
Conversely, for any $\left(R_X,R_Y\right)$ outside of $\mathcal{C}$,
the average probability of error is bounded away from 0. The set
$\mathcal{C}$, called \emph{capacity region} for $W$, is the closure
of the set of all rate pairs $\left(R_X,R_Y\right)$
satisfying~\cite{SlWo73}
\begin{subequations}
\begin{align}
0 &\leq R_X \leq I\left(X \wedge Z|Y,U\right)\\
0 &\leq R_Y \leq I\left(Y \wedge Z|X,U\right)\\
0 &\leq R_X+R_Y \leq I\left(XY \wedge Z|U\right),
\end{align} 
\end{subequations}
%
for all choices of joint distributions over the random variables
$U,\ X,\ Y,\ Z$ of the form
$p\left(u\right)p\left(x|u\right)p\left(y|u\right)W\left(z|x,y\right)$
with $U \in \U$ and $|\mathcal{U}| \leq 4$. As we can see, this
theorem was presented in an asymptotic nature, i.e., it was proven
that the error probability of the channel code can go to zero as the
block length goes to infinity. It does not tell us how large the
block length must be in order to achieve a specific error
probability. On the other hand, in practical situations, there are
limitations on the delay of the communication. Additionally, the
block length of the code cannot go to infinity. Therefore, it is
important to study how the probability of error drops as the block
length goes to infinity. A partial answer to this question is
provided by examining the error exponent of the channel.

Error exponents have been studied for discrete memoryless
multiple-access channels over the past thirty years. Lower and upper
bounds are known on the error exponent of these channels. The random
coding bound in information theory provides a well-known lower bound
for the reliability function of the best code, of a given rate and
block length. This bound is constructed by upper-bounding the
average error probability over an ensemble of codes. Slepian and
Wolf~\cite{SlWo73}, Dyachkov~\cite{Dyachkov},
Gallager~\cite{Gallager-Multiaccess}, Pokorny and
Wallmeier~\cite{Pokorney}, and Liu and
Hughes~\cite{Liu-RandomCoding} have all studied the random coding
bound for discrete memoryless multiple access channels. Nazari and
et al.~\cite{aransspaa09} investigated two different upper bounds on
the average probability of error, called the typical random coding
bound and the partial expurgated bound. The typical bound is
basically the typical performance of the ensemble. By this, we mean
that almost all random codes exhibit this performance. In addition,
they have shown that the typical random code performs better than
the average performance over the random coding ensemble, at least,
at low rates. The random coding exponent may be improved at low
rates by a process called ``partial expurgation'' which yields a new
bound that exceeds the random coding bound at low rates.

Haroutunian~\cite{Haroutunian} and
Nazari~\cite{nazari08,nazaridistance09} studied upper bounds on the
error exponent of multiple access channels. In Multi-user
information theory, the sphere packing bound provides a well known
upper bound on the reliability function for multiple access channel.
The sphere packing bound that Haroutunian~\cite{Haroutunian} derived
on the average error exponent for DM-MAC is potentially loose, as it
does not capture the separation of the encoders in the MAC. Nazari
et al.~\cite{nazari08} derived another sphere packing bound which
takes into account separation of the encoders. The bound
in~\cite{nazari08} turns out to be at least as good as the bound
derived in~\cite{Haroutunian}, however it is a valid bound only for
the maximal error exponent and not the average. The sphere packing
bound is a good bound in high rate regime. Nevertheless, it tends to
be a loose bound in low rate regime. It can be shown that in low
rate regime, the minimum distance of the code dominates the
probability of error. Using the minimum distance of the code,
Nazari~\cite{nazaridistance09} derived another upper bound for the
maximal error exponent of DM-MAC. To derive the minimum distance
bound, they established a connection between the minimum distance of
the code and the maximum probability of error; then, by obtaining an
upper bound on the minimum distance of all codes with certain rates,
they derived a lower bound on the maximal error probability that can
be obtained by a code with a certain rate pair.

The paper is organized as follows. Some preliminaries are introduced
in section~\ref{prelim}. The main result of the paper, which is an
upper bound on the reliability function of the channel, is obtained
in section~\ref{SpherePack}. In section~\ref{MaxvsAvg}, by using a
known upper bound on the maximum error exponent function, we derive
an upper bound on the average error exponent function. The proofs of
some of these results are given in the Appendix.

\section{Preliminaries}\label{prelim}

For any alphabet $\mathcal{X}$, $\mathcal{P\left(X\right)}$ denotes
the set of all probability distributions on $\mathcal{X}$. The
\emph{type} of a sequence $\mathbf{x}=\left(x_1,...,x_n\right) \in
\mathcal{X}^n$ is the distributions $P_{\mathbf{x}}$, on
$\mathcal{X}$, defined by:
\begin{align}
P_{\mathbf{x}}\left(x\right)\triangleq
\frac{1}{n}N\left(x|\mathbf{x}\right),\;\;\;\;\;\;\;\;\;\;\;\;\;\;\;
x \in \mathcal{X},
\end{align}
where $N\left(x|\mathbf{x}\right)$ denotes the number of occurrences
of $x$ in $\mathbf{x}$. Let $\mathcal{P}_n \left(\mathcal{X}\right)$
denotes the set of all types in $\mathcal{X}^n$, and define the set
of all sequences in $\mathcal{X}^n$ of type $P$ as
\begin{align}
T_P \triangleq \{\mathbf{x} \in \mathcal{X}^n: P_{\mathbf{x}}=P\}.
\end{align}
The joint type of a pair $\left(\mathbf{x},\mathbf{y}\right) \in
\mathcal{X}^n \times \mathcal{Y}^n$ is the probability distribution
$P_{\mathbf{x},\mathbf{y}}$ on $\mathcal{X} \times \mathcal{Y}$
defined by:
\begin{align}
P_{\mathbf{x},\mathbf{y}}\left(x,y\right)\triangleq
\frac{1}{n}N\left(x,y|\mathbf{x},\mathbf{y}\right),\;\;\;\;\;\;\;\;\;\;\;\;\;\;\;
\left(x,y\right) \in \mathcal{X} \times \mathcal{Y},
\end{align}
where $N\left(x,y|\mathbf{x},\mathbf{y}\right)$ is the number of
occurrences of $\left(x,y\right)$ in ($\mathbf{x},\mathbf{y}$). The
relative entropy or \emph{Kullback-Leibler} distance between two
probability distribution $P,Q \; \in \mathcal{P\left(X\right)}$ is
defined as
\begin{align}
D\left(P||Q\right) \triangleq \sum_{x \in
\mathcal{X}}P\left(x\right)\log\frac{P\left(x\right)}{Q\left(x\right)}.
\end{align}
Let $\mathcal{W\left(Y|X\right)}$ denote the set of all stochastic
matrices with input alphabet $\mathcal{X}$ and output alphabet
$\mathcal{Y}$. Then, given stochastic matrices $V,\ W \in
\mathcal{W\left(Y|X\right)}$, the conditional \emph{I-divergence} is
defined by
\begin{align}
D\left(V||W|P\right) \triangleq \sum_{x \in
\mathcal{X}}P\left(x\right)D\left(V\left(\cdot|x\right)||W\left(\cdot|x\right)\right).
\end{align}

\begin{definition}
An $\left(n,M,N\right)$ multi-user code is a set
$\{\left(\mathbf{x}_i,\mathbf{y}_j,D_{ij}\right): 1 \leq i \leq M, 1
\leq j \leq N\}$ with
\begin{itemize}
\item $\mathbf{x}_i \in \mathcal{X}^n$, $\mathbf{y}_j \in
\mathcal{Y}^n$, $D_{ij} \subset \mathcal{Z}^n$
\item $D_{ij} \cap D_{i'j'}=\varnothing$ for  $\left(i,j\right) \neq
\left(i',j'\right)$.
\end{itemize}
The average error probability of this code for the MAC, $W:\X \times
\Y \rightarrow \Z$, is defined as
\begin{align}
e\left(\C,W\right) \triangleq \frac{1}{M
N}\sum_{i=1}^{M}\sum_{j=1}^{N}W^n\left(D_{i,j}^c|\mathbf{x}_i,\mathbf{y}_j\right).
\end{align}
Similarly, the maximal error probability of this code for $W$ is
defined as
\begin{align}
e_m\left(\C,W\right) \triangleq \max_{\substack{\left(i,j\right)}}
W^n\left(D_{i,j}^c|\mathbf{x}_i,\mathbf{y}_j\right).
\end{align}
\end{definition}
\begin{definition}
For the MAC, $W:\X \times \Y \rightarrow \Z$, the average and
maximal error reliability functions, at rate pair
$\left(R_X,R_Y\right)$, are defined as:
\begin{align}
&E^*_{av}\left(R_X,R_Y\right) \triangleq \lim_{n\rightarrow \infty}
\max_{\substack{\C}} -\frac{1}{n}
\log{e\left(\C,W\right)}\\
&E^*_{m}\left(R_X,R_Y\right) \triangleq \lim_{n\rightarrow \infty}
\max_{\substack{\C}} -\frac{1}{n} \log{e_m\left(\C,W\right)},
\end{align}
where the maximum is over all codes of length $n$ and rate pair
$\left(R_X,R_Y\right)$.
\end{definition}

\begin{definition}
 A code $\C_X =\{ \mathbf{x}_i \in \X^n
:\;\; i=1,...,M_X \}$, for some $P_X$, is called a bad codebook, if
\begin{eqnarray}
\exists \;\left(i,j\right),\;\;\;\; i \neq j \;\;\; &
\mathbf{x}_i=\mathbf{x}_j
\end{eqnarray}
A codebook which is not bad, is called a good one.
\end{definition}
\begin{definition}
 A multi user code $\C=\C_X \times \C_Y$ is
called a good multi user code, if both individual codebooks $\C_X$,
$\C_Y$ are good codes.
\end{definition}
\begin{definition}
For a good multi user code $\C=\C_X \times \C_Y$, and for a
particular type $P_{XY} \in \mathcal{P}_n\left(\mathcal{X}\times
\mathcal{Y}\right)$, we define
\begin{eqnarray}
R\left(\C,P_{XY}\right) \triangleq \frac{1}{n} \log|\C \cap
T_{P_{XY}}|
\end{eqnarray}
\end{definition}

\begin{definition}
For a sequence of joint types $P^n_{XY} \in \P_n\left(\X \times
\Y\right)$, with marginal types $P^n_X$ and $P^n_Y$, the sequence of
type graphs, $G_n$, is defined as follows: For every $n$, $G_n$ is a
bipartite graph, with its left vertices consisting of all $x^n \in
T_{P^n_X}$ and the right vertices consisting of  all $y^n \in
T_{P^n_Y}$. A vertex on the left (say $\tilde{x}^n$) is connected to
a vertex on the right (say $\tilde{y}^n$) if and only if
$\left(\tilde{x}^n,\tilde{y}^n\right) \in T_{P^n_{XY}}$.
\end{definition}

\section{main result}\label{SpherePack}
The main result of this section is a new sphere packing bound for
the average error probability for a discrete memoryless multiple
access channel. The idea behind the derivation of this bound is
based on the property that is common among all good multi user codes
with certain rate pair. In the following, we first derive a sphere
packing bound for a good multiuser code. Next, we show that for any
bad multiuser code, there exists a good code with the same rate pair
and smaller average probability of error. Therefore, to obtain a
lower bound for the average error probability for the best code, we
only need to study good codes (codes without any repeated
codewords).

Now, consider a good multiuser code with blocklength $n$. Suppose
the number of messages of the first source is $M_X=2^{nR_X}$ and the
number of messages of the second source is $M_Y=2^{nR_Y}$. Assume
that all the messages of any source are equiprobable and the sources
are sending data independently. Considering these assumptions, all
$M_XM_Y$ pairs are occuring with the equal probability. Thus, at the
input of the channel, we can see all possible
$2^{n\left(R_X+R_Y\right)}$ (an exponentially increasing function of
$n$) pairs of input sequences. However, we also know that the number
of possible types is a polynomial function of $n$. Thus, for at
least one joint type, the number of pairs of sequences in the multi
user code sharing that particular type, should be an exponential
function of $n$ with the rate arbitrary close to the rate of the
multi user code. We will look at these pairs of sequences as a
subcode, and then try to find a lower bound for the average error
probability of this subcode. Following, we will show that this bound
is a valid lower bound for the average probability of error for the
original code.

\begin{lemma}~\cite{nazaridistance09}
For any $\delta > 0$, for any sufficiently large $n$, and for any
good $\left(n,2^{nR_X},2^{NR_Y}\right)$ multi user code $\C$, as
defined above, there exists $P_{XY} \in
\mathcal{P}_n\left(\mathcal{X}\times \mathcal{Y}\right)$ such that
\begin{eqnarray*}
R\left(\C,P_{XY}\right) \geq
R_X+R_Y-\delta\;\;\;\;\;\;\;\;\;\;\;\;\; \text{for sufficiently
large n},
\end{eqnarray*}
$P_{XY}$ is called a dominant type of $\C$.
\end{lemma}

Hence, for any good code, there must exist at least a joint type
which dominates the codebook. We can ask the following question: for
a multiuser code, with rate $\left(R_X,R_Y\right)$, can any joint
type potentially be its dominant type? As shown later, the answer to
this question helps us characterize a tighter sphere packing bound.
In response to this question, Nazari et al.~\cite{nazaridistance09}
studied the type graphs for different joint types and proved the
following result:

\begin{lemma}\cite{nazaridistance09} \label{U-Lemma}
For all sequences of nearly complete subgraphs of a particular type
graph $T_{P_{XY}}$, the rates of the subgraph $\left(R_X,R_Y\right)$
must satisfy
\begin{eqnarray}
 R_X \leq H\left(X|U\right), \: R_Y \leq H\left(Y|U\right)
\end{eqnarray}
for some $P_{U|XY}$ such that $X-U-Y$.
\end{lemma}

Now consider a particular joint type $P^n_{XY}$. By the previous
lemma, if there  does not exist any $P_{U|XY}$ satisfying the
constraint mentioned in lemma~\ref{U-Lemma}, the type graph
corresponding to this joint type can not contain an almost fully
connected subgraph with rate $\left(R_X,R_Y\right)$. Consequently,
it cannot be the dominant type of a good multiuser code with rate
$\left(R_X,R_Y\right)$.

\begin{fact} \label{U2-Lemma}
Consider a good multiuser code $\C$ with parameter
$\left(n,2^{nR_X},2^{nR_Y}\right)$. A joint type $P^n_{XY} \in
\P_n\left(\X \times \Y\right)$ can be the dominant type of $\C$ if
there exists a $P_{U|XY}$, $X-U-Y$, such that
\begin{eqnarray}
R_X \leq H\left(X|U\right), \: R_Y \leq H\left(Y|U\right),
\end{eqnarray}
conversely, if it does not exist such a conditional distribution,
then $P^n_{XY}$ cannot be the dominant type of any good multiuser
code with parameter $\left(n,2^{nR_X},2^{nR_Y}\right)$.
\end{fact}

\begin{theorem}\label{fixedtypepacking}
Fix any $R_X \geq 0$, $R_Y \geq 0$, $\delta > 0$ and a sufficiently
large $n$. Consider a good multiuser code $\C$ with parameter
$\left(n,2^{nR_X},2^{nR_Y}\right)$ which has a dominant type
$P_{XY}^* \in \mathcal{P}_n \left(\mathcal{X} \times
\mathcal{Y}\right)$. The average error exponent of such a code is
bounded above by
\begin{eqnarray}
E_{sp}\left(R_X,R_Y,W\right) \triangleq \min_{V_{Z|XY}}
D\left(V_{Z|XY}||W|P^*_{XY}\right).
\end{eqnarray}
Here, the minimization is over all possible conditional
distributions $V_{Z|XY}:\X \times \Y \rightarrow \Z$, which satisfy
at least one of the following conditions
\begin{eqnarray}
I_V\left(X \wedge Z|Y\right) &\leq& R_X\\
I_V\left(Y \wedge Z|X\right) &\leq& R_Y\\
I_V\left(XY \wedge Z\right) &\leq& R_X+R_Y.
\end{eqnarray}

\end{theorem}
\begin{proof}
The proof is provided in Appendix~A.1.
\end{proof}

In theorem~\ref{fixedtypepacking}, we obtained a sphere packing
bound on the average error exponent for a good multiuser code with a
certain dominant type. For a more general code, we do not know the
dominant type of the code. However, we do have the condition for a
joint type to be the potential dominant type of a code with certain
parameter. By combining the result of theorem~\ref{fixedtypepacking}
and fact~\ref{U2-Lemma}, we can obtain the following sphere packing
bound for any good multiuser code:

\begin{theorem}\label{packing}
For any given multiple access channel $W$ and any good multi user
code with rate pair ($R_X$,$R_Y$), the reliability function,
$E\left(R_X,R_Y,W\right)$, is bounded above by
\begin{eqnarray}
E_{sp}\left(R_X,R_Y,W\right) \triangleq \max_{P_{UXY}}
\min_{V_{Z|XY}} D\left(V_{Z|XY}||W|P_{XY}\right).
\end{eqnarray}
Here, the maximum is taken over all possible joint distributions
satisfying $X-U-Y$ and
\begin{eqnarray}
 R_X \leq H\left(X|U\right), \: R_Y \leq H\left(Y|U\right),
\end{eqnarray}
 and the minimum over all channels $V_{Z|XY}$ that satisfy at
least one of the following conditions
\begin{eqnarray}
I_V\left(X \wedge Z|Y\right) &\leq& R_X\\
I_V\left(Y \wedge Z|X\right) &\leq& R_Y\\
I_V\left(XY \wedge Z\right) &\leq& R_X+R_Y.
\end{eqnarray}
\end{theorem}

Thus far, we have obtained a lower bound on the average error
probability for all good multiuser codes with certain rate pairs.
Here, we show that the result of the previous theorem is indeed a
valid bound for any multiuser code regardless of whether it is good
or bad. This approach shows that for any bad code there exists a
good code with the same number of codewords and a better
performance. Therefore, to obtain a lower bound on the error
probability of the best code, we only need to consider codes without
any repeated codewords. In lemma~\ref{P2P}, we prove this result for
a single-user code and later, by using the result of lemma~\ref{P2P}
several times, we prove the same result for the multiuser scenario.

\begin{lemma}\label{P2P}
Suppose $C_X$ is a codebook of size $M_X$ for which all codewords
are selected from $T_{P_X}$. Moreover, suppose $\mathbf{x}_i$ is
repeated $N_i$ times in the codebook and $M_X=N_1+N_2+...+N_M$,
where M is the number of distinct sequences in $C_X$. If $M_X \leq
|T_{P_X}|-1$, there exists another code $C'_X$ with better
probability of error, such that
\begin{eqnarray}
|C_X|&=&|C'_X|\nonumber\\
N'_i&=&N_i\;\;\;\;\;\;\;  i=1,...,M-1\nonumber\\
N'_M&=&N_M-1\nonumber\\
N'_{M+1}&=&1
\end{eqnarray}
Here, $N'_{M+1}=1$ is the number of occurrences of the new sequence
$\mathbf{x} \in T_{P_X}$ which does not belong to $C_X$.
\end{lemma}
\begin{proof}
The proof is provided in Appendix~A.2.
\end{proof}

\begin{lemma}\label{MAC}
For any bad multi user code with codewords that belong to $T_{P_X}$,
and $T_{P_Y}$, with rate pair ($R_X$,$R_Y$), there exists a good
multi user code with the same rate pair and a better probability of
error.
\end{lemma}
\begin{proof}
For a bad multi user code, we know that at least one of the
individual codebooks is bad. If we apply lemma~\ref{P2P} several
times to any of the bad single user codes, with the appropriate
cardinality, we will end up with a good multiuser code and a better
probability of decoding error.
\end{proof}

Finally, by combining the result of lemma~\ref{MAC} and the result
of theorem~\ref{packing}, we deduce an upper bound on the
reliability function for all multiuser codes.

\begin{theorem}
For any given multiple access channel W, and any good multi user
code with rate pair ($R_X$,$R_Y$), the reliability function,
$E\left(R_X,R_Y,W\right)$, is bounded above by
\begin{eqnarray}
E_{sp}\left(R_X,R_Y,W\right) \triangleq \max_{P_{XY}}
\min_{V_{Z|XY}} D\left(V_{Z|XY}||W|P_{XY}\right).
\end{eqnarray}
 Here, the maximum is taken over all possible joint
distributions, and the minimum over all channels $V_{Z|XY}$ which
satisfy at least one of the following conditions
\begin{eqnarray}
I_V\left(X \wedge Z|Y\right) &\leq& R_X\nonumber\\
I_V\left(Y \wedge Z|X\right) &\leq& R_Y\nonumber\\
I_V\left(XY \wedge Z\right) &\leq& R_X+R_Y
\end{eqnarray}
\end{theorem}

\section{Another Sphere packing bound}\label{MaxvsAvg}
In point to point communications systems, one can show that a lower
bound for the maximal error probability of the best code is also a
lower bound on the average probability of error for such a code.
However, in multiuser communications, this is not the case. It has
been shown that for multiuser channels, in general, the maximal
error capacity region is smaller than the average error capacity
region~\cite{Dueck2}. Therefore, we cannot hope a sphere packing
bound for maximal error probability to be equal to the one for the
average probability of error. In the following, we show an approach
to derive an upper bound on the average error exponent by using a
known upper bound for the maximal error exponent.

\begin{lemma}
Fix any DM-MAC $W:\X \times \Y \rightarrow \Z$, $R_X \geq 0$, $R_Y
\geq 0$. Assume that, the maximal reliability function is bounded as
follows:
\begin{equation}
E^L_{m}\left(R_X,R_Y\right) \leq E^*_{m}\left(R_X,R_Y\right) \leq
E^U_{m}\left(R_X,R_Y\right),\label{number1}
\end{equation}
therefore, the average reliability function can be bounded by
\begin{equation}
E^L_{m}\left(R_X,R_Y\right) \leq E^*_{av}\left(R_X,R_Y\right) \leq
E^U_{m}\left(R_X,R_Y\right) + R,\label{number2}
\end{equation}
where $R = \min\{R_X,R_Y\}$. Similarly, if the average reliability
function is bounded as follows:
\begin{equation}
E^L_{av}\left(R_X,R_Y\right) \leq E^*_{av}\left(R_X,R_Y\right) \leq
E^U_{av}\left(R_X,R_Y\right),\label{number3}
\end{equation}
it can be concluded that the maximal reliability function satisfies
the following constraint
\begin{equation}
E^L_{av}\left(R_X,R_Y\right) -R \leq E^*_{m}\left(R_X,R_Y\right)
\leq E^U_{av}\left(R_X,R_Y\right).\label{number4}
\end{equation}
\end{lemma}
\begin{proof}
The proof is provided in Appendix~A.3.
\end{proof}
In~\cite{nazari08}, the authors derived a sphere packing bound on
the maximal reliability function for DM-MAC. This results is only a
valid upper bound for the maximal error reliability function and not
the average one. We can simply use the previous lemma to derive a
new upper bound on the average error reliability function for
DM-MAC.
\begin{theorem}
For any $R_X,R_Y
>0$, $\delta > 0$ and any DM-MAC, $W:\mathcal{X} \times \mathcal{Y}\rightarrow
\mathcal{Z}$, every $(n,M_X,M_Y)$ code, $\C$  with
\begin{subequations}
\begin{align}
\frac{1}{n} \log{M_X} &\geq R_X + \delta \\
\frac{1}{n} \log{M_Y} &\geq R_Y + \delta,
\end{align} 
\end{subequations}
has average probability of error
\begin{equation}
e(\C,W) \geq \frac{1}{2}\exp{\left(-n
\left(E^m_{sp}(R_X,R_Y,W)+R\right)\left(1+\delta\right)\right)},
\end{equation}
where $E^m_{sp}$ is the sphere packing bound derived
in~\cite{nazari08}, and $R=\min\{R_X,R_Y\}$.
\end{theorem}


\section{Appendix}
\subsection{Appendix A.1}
For a given MAC $W:\mathcal{X} \times \mathcal{Y} \rightarrow
\mathcal{Z}$ and a  good multi user code $C=C_X \times C_Y$, where
$C_X =\{\ \textbf{x}_i \in \X^n :\;\; i=1,...,M_X\}$ and $C_Y =\{\
\textbf{y}_j \in \Y^n :\;\; j=1,...,M_Y\}$, with decoding sets
$D_{i,j} \subset \mathcal{Z}^n$, we have
\begin{eqnarray}
e\left(C,W\right)&=&\frac{1}{M_XM_Y}\sum_{i=1}^{M_X}\sum_{j=1}^{M_Y}W\left(D_{i,j}^c|\mathbf{x}_i,\mathbf{x}_j\right)\\
&=&\frac{1}{M_XM_Y}\sum_{P_{XY}}
\frac{M_{XY}}{M_{XY}}\sum_{\left(i,j\right)\in
C_{XY}}W\left(D_{i,j}^c|\mathbf{x}_i,\mathbf{x}_j\right)
\end{eqnarray}
where $\C_{XY}$ is the set that includes all pairs in $\C_X\times
\C_Y$ which have the same type $P_{XY}$, $M_{XY}$ denotes the
cardinality of this set, and $R_{XY}=\frac{1}{n} \log{M_{XY}}$. For
a fixed $\left(i,j\right)$,
$T_V\left(\mathbf{x}_i,\mathbf{x}_j\right)$s are disjoint subsets of
$\Z^n$ for different conditional types $V: \X \times \Y \rightarrow
\Z$. Therefore,
\begin{eqnarray}
e\left(C,W\right)&=&\frac{1}{M_XM_Y}\sum_{P_{XY}}\frac{M_{XY}}{M_{XY}}\sum_{\left(i,j\right)\in
C_{XY}}\sum_{V}W\left(D_{i,j}^c \cap
T_V\left(\mathbf{x}_i,\mathbf{y}_j\right)|\mathbf{x}_i,\mathbf{y}_j\right)\\
&=&\frac{1}{M_XM_Y}\sum_{P_{XY}}
M_{XY}\sum_{V}\frac{1}{M_{XY}}\sum_{\left(i,j\right)\in
C_{XY}}W\left(T_V\left(\mathbf{x}_i,\mathbf{y}_j\right)|i,j\right)\frac{|D_{i,j}^c \cap T_V\left(\mathbf{x}_i,\mathbf{y}_j\right)|}{|T_V\left(\mathbf{x}_i,\mathbf{y}_j\right)|}\\
&=& \frac{1}{M_XM_Y}\sum_{P_{XY}} M_{XY}\sum_{V}
2^{-nD\left(V||W|P_{XY}\right)}[1-\frac{1}{M_{XY}}\sum_{\left(i,j\right)\in
C_{XY}}\frac{|D_{i,j} \cap
T_V\left(\mathbf{x}_i,\mathbf{y}_j\right)|}{|T_V\left(\mathbf{x}_i,\mathbf{y}_j\right)|}]\\
&\geq& \frac{1}{M_XM_Y}\sum_{P_{XY}} M_{XY}\sum_{V}
2^{-nD\left(V||W|P_{XY}\right)}[1-\frac{1}{M_{XY}}\sum_{\left(i,j\right)\in
C_{XY}}\frac{|D_{i,j} \cap
T_V\left(\mathbf{x}_i,\mathbf{y}_j\right)|}{2^{nH\left(Z|X,Y\right)}}]\\
&=& \frac{1}{M_XM_Y}\sum_{P_{XY}} M_{XY}\sum_{V}
2^{-nD\left(V||W|P_{XY}\right)}[1-\frac{1}{M_{XY}}
\frac{|\bigcup_{\left(i,j\right)
\in C_{XY}} D_{i,j}\cap T_V\left(\mathbf{x}_i,\mathbf{y}_j\right)|}{2^{nH\left(Z|X,Y\right)}}]\\
&\geq& \frac{1}{M_XM_Y}\sum_{P_{XY}} M_{XY}\sum_{V}
2^{-nD\left(V||W|P_{XY}\right)}[1-\frac{1}{M_{XY}}
\frac{|T_Z|}{2^{nH\left(Z|X,Y\right)}}]\\
&\geq&\frac{1}{M_XM_Y}\sum_{P_{XY}} M_{XY}\sum_{V}
2^{-nD\left(V||W|P_{XY}\right)}[1-\frac{1}{M_{XY}}
\frac{2^{nH\left(Z\right)}}{2^{nH\left(Z|X,Y\right)}}]\\
&=& \frac{1}{M_XM_Y}\sum_{P_{XY}} M_{XY}\sum_{V}
2^{-nD\left(V||W|P_{XY}\right)}[1-2^{-n[R_{XY}-I_V\left(XY \wedge
Z\right)]}]
\end{eqnarray}


 We define
\begin{eqnarray}
V_{bad}^{XY}=\{ V : R_{XY}\geq I_V\left(XY \wedge Z\right)\}
\end{eqnarray}
So, form the last inequality,
\begin{eqnarray}
e\left(C,W\right) &\geq& \frac{1}{M_X M_Y} \sum_{P_{XY}}M_{XY}
\sum_{V \in
V_{bad}^{XY}} 2^{-n D\left(V||W|P_{XY}\right)}\\
&\geq& \frac{1}{M_X M_Y} \sum_{P_{XY}}M_{XY} 2^{-n[\min_{V \in
V_{bad}^{XY}}D\left(V||W|P_{XY}\right)]}\\
&=& \frac{1}{M_X M_Y} \sum_{P_{XY}} 2^{-n[\min_{V \in
V_{bad}^{XY}}D\left(V||W|P_{XY}\right)-R_{XY}]}\\
&\geq& \frac{1}{M_X M_Y} 2^{-n[\min_{P_{XY}}\min_{V \in
V_{bad}^{XY}}D\left(V||W|P_{XY}\right)-R_{XY}]}
\end{eqnarray}
Thus,
\begin{equation}
e\left(C,W\right) \geq 2^{-n[\min_{P_{XY}}\min_{V \in
V_{bad}^{XY}}D\left(V||W|P_{XY}\right)+R_X+R_Y-R_{XY}]}
\end{equation}
On the other hand, if we use the fact that $D_{i,j}^c \subseteq
\bigcup_{j'}\bigcup_{i'\neq i} D_{i',j'}$, we can conclude
\begin{eqnarray}
e\left(C,W\right)&=&\frac{1}{M_XM_Y}\sum_{i=1}^{M_X}\sum_{j=1}^{M_Y}W\left(D_{i,j}^c|\mathbf{x}_i,\mathbf{y}_j\right)\\
&\geq&  \frac{1}{M_Y}\sum_{j=1}^{M_Y} \frac{1}{M_X}\sum_{i=1}^{M_X}
W\left(\bigcup_{j'}\bigcup_{i'\neq i} D_{i',j'}|\mathbf{x}_i,\mathbf{y}_j\right)\\
\text{Define} D_i^c\triangleq \bigcup_{j'}\bigcup_{i'\neq i} D_{i',j'}\\
&=& \frac{1}{M_XM_Y}\sum_{P_{XY}}\sum_{i}\sum_{j:\left(i,j\right)\in
C_{XY}}\sum_{V}W\left(D_i^c \cap T_V\left(\mathbf{x}_i,\mathbf{y}_j\right)|\mathbf{x}_i,\mathbf{y}_j\right)\\
&=& \frac{1}{M_XM_Y}\sum_{P_{XY}}\sum_{V} 2^{-n
D\left(V||W|P_{XY}\right) }\sum_{i}\sum_{j:\left(i,j\right)\in
C_{XY}} \frac{|D_i^c \cap T_V\left(\mathbf{x}_i,\mathbf{y}_j\right)|}{|T_V\left(\mathbf{x}_i,\mathbf{y}_j\right)|}\\
&=& \sum_{P_{XY}}\sum_{V} 2^{-n D\left(V||W|P_{XY}\right)}
\frac{1}{M_XM_Y} \sum_{i}\sum_{j:\left(i,j\right)\in
C_{XY}}[1-\frac{|D_i \cap
T_V\left(\mathbf{x}_i,\mathbf{y}_j\right)|}{|T_V\left(\mathbf{x}_i,\mathbf{y}_j\right)|}]\\
&=& \sum_{P_{XY}}\sum_{V} 2^{-n D\left(V||W|P_{XY}\right)}
\frac{M_{XY}}{M_XM_Y} [1-\frac{1}{M_{XY}}
\sum_{i}\sum_{j:\left(i,j\right)\in
C_{XY}} \frac{|D_i \cap T_V\left(\mathbf{x}_i,\mathbf{y}_j\right)|}{|T_V\left(\mathbf{x}_i,\mathbf{y}_j\right)|}]\;\;\;\;\;\;\;\;\;\\
&\geq& \sum_{P_{XY}}\sum_{V} 2^{-n D\left(V||W|P_{XY}\right)}
\frac{M_{XY}}{M_XM_Y} [1-\frac{1}{M_{XY}}\sum_{i=1}^{M_X}
\sum_{j=1}^{M_Y} \frac{|D_i \cap T_V\left(\mathbf{x}_i,\mathbf{y}_j\right)|}{|T_V\left(\mathbf{x}_i,\mathbf{y}_j\right)|}]\\
&\geq&\sum_{P_{XY}}\sum_{V} 2^{-n D\left(V||W|P_{XY}\right)}
\frac{M_{XY}}{M_XM_Y} [1-\frac{1}{M_{XY}} \sum_{j=1}^{M_Y}
\sum_{i=1}^{M_X} \frac{|D_i
\cap T_V\left(\mathbf{x}_i,\mathbf{y}_j\right)|}{2^{nH\left(Z|X,Y\right)}}]\\
&\geq& \sum_{P_{XY}}\sum_{V} 2^{-n D\left(V||W|P_{XY}\right)}
\frac{M_{XY}}{M_XM_Y} [1-\frac{1}{M_{XY}} \sum_{j=1}^{M_Y}
\frac{2^{nH\left(Z,X|Y\right)}}{2^{nH\left(Z|X,Y\right)}}]\\
&\geq&\sum_{P_{XY}}\sum_{V} 2^{-n D\left(V||W|P_{XY}\right)}
\frac{M_{XY}}{M_XM_Y} [1-2^{-n[R_{XY}-R_Y-I_V\left(Z\wedge
X|Y\right)]}]
\end{eqnarray}
 and now, let us define
\begin{eqnarray}
V_{bad}^{X}\triangleq \{ V : R_{XY}-R_Y\geq I_V\left(Z\wedge
X|Y\right)\}
\end{eqnarray}
Hence, it easily can be seen
\begin{eqnarray}
e\left(C,W\right) &\geq& \frac{1}{M_X M_Y} \sum_{P_{XY}}M_{XY}
\sum_{V \in
V_{bad}^{X}} 2^{-n D\left(V||W|P_{XY}\right)}\\
&\geq&\frac{1}{M_X M_Y} \sum_{P_{XY}}M_{XY} 2^{-n[\min_{V \in
V_{bad}^{X}}D\left(V||W|P_{XY}\right)]}\\
&=& \frac{1}{M_X M_Y} \sum_{P_{XY}} 2^{-n[\min_{V \in
V_{bad}^{X}}D\left(V||W|P_{XY}\right)-R_{XY}]}\\
&\geq& \frac{1}{M_X M_Y} 2^{-n[\min_{P_{XY}}\min_{V \in
V_{bad}^{X}}D\left(V||W|P_{XY}\right)-R_{XY}]}
\end{eqnarray}
So,
\begin{equation}
e\left(C,W\right) \geq 2^{-n[\min_{P_{XY}}\min_{V \in
V_{bad}^{X}}D\left(V||W|P_{XY}\right)+R_X+R_Y-R_{XY}]}
\end{equation}
 Using the same idea for $Y$ and defining $D_j^c\triangleq
\bigcup_{i'}\bigcup_{j'\neq j} D_{i',j'}$, we can easily see
\begin{eqnarray}
e\left(C,W\right) &\geq& \frac{1}{M_X M_Y} \sum_{P_{XY}}M_{XY}
\sum_{V \in
V_{bad}^{Y}} 2^{-n D\left(V||W|P_{XY}\right)}\\
&\geq& \frac{1}{M_X M_Y} \sum_{P_{XY}}M_{XY} 2^{-n[\min_{V \in
V_{bad}^{Y}}D\left(V||W|P_{XY}\right)]}\\
&=& \frac{1}{M_X M_Y} \sum_{P_{XY}} 2^{-n[\min_{V \in
V_{bad}^{Y}}D\left(V||W|P_{XY}\right)-R_{XY}]}\\
&\geq& \frac{1}{M_X M_Y} 2^{-n[\min_{P_{XY}}\min_{V \in
V_{bad}^{Y}}D\left(V||W|P_{XY}\right)-R_{XY}]}
\end{eqnarray}
So,
\begin{equation}
e\left(C,W\right) \geq 2^{-n[\min_{P_{XY}}\min_{V \in
V_{bad}^{Y}}D\left(V||W|P_{XY}\right)+R_X+R_Y-R_{XY}]}
\end{equation}
where
\begin{eqnarray}
V_{bad}^{Y}=\{ V : R_{XY}-R_X\geq I_V\left(Z\wedge Y|X\right)\}
\end{eqnarray}
From (64),(81),(86),
\begin{eqnarray}
e\left(C,W\right) \geq 2^{-n[ \min_{P_{XY}} \min_{V \in V_{bad}^{X}
\cup V_{bad}^{Y}\cup V_{bad}^{XY}
}D\left(V||W|P_{XY}\right)+R_X+R_Y-R_{XY}]}.
\end{eqnarray}
Equivalently, for the exponent of $e\left(C,W\right)$
\begin{equation}
E\left(C,W\right) \leq \min_{P_{XY}} \min_{V \in V_{bad}^{X} \cup
V_{bad}^{Y}\cup V_{bad}^{XY}} D\left(V \|W
|P_{XY}\right)+R_X+R_Y-R_{XY}
\end{equation}
If we define $V_{bad}=V_{bad}^{X} \cup V_{bad}^{Y}\cup V_{bad}^{XY}
$, for every code $C$, we have
\begin{eqnarray}
E\left(C,W\right) &\leq& \max_C \;\;\;\;\;\min_{P_{XY}} \min_{V \in
V_{bad}} D\left(V \|W
|P_{XY}\right)+R_X+R_Y-R_{XY}\\
&=& \max_{\underline{R} \in \mathcal{R}} \;\;\;\;\;\min_{P_{XY}}
\min_{V \in V_{bad}} D\left(V \|W |P_{XY}\right)+R_X+R_Y-R_{XY}
\end{eqnarray}
Where $\underline{R}$ is a vector with elements
$R\left(C,P_{XY}\right)$ and $\mathcal{R}$ is the set of all
possible vectors $\underline{R}$. The last inequality follows from
the fact that $E\left(C,W\right)$ is only a function of $R_{XY}$s.
Since $P_{XY}^*$ is the dominant type of the code, we conclude that
\begin{eqnarray}
E\left(C,W\right) &\leq& \max_{\underline{R} \in \mathcal{R}} \;\;
\min_{V \in
V_{bad}} D\left(V \|W |P_{XY}^*\right)+R_X+R_Y-R_{XY}^*\\
&=& \max_{\underline{R} \in \mathcal{R}} \;\; \min_{V \in V_{bad}}
D\left(V \|W |P_{XY}^*\right).
\end{eqnarray}
However, this expression does not depend on $\underline{R}$.
Therefore
\begin{eqnarray}
E\left(C,W\right) &\leq& \min_{V \in V_{bad}} D\left(V \|W
|P_{XY}^*\right),
\end{eqnarray}
where $V_{bad}=\{V:I_V\left(XY\wedge Z\right) \leq R_X+R_Y \;\; or
\;\; I_V\left(Y\wedge Z| X\right) \leq R_Y \;\;or \;\;
I_V\left(X\wedge Z |Y\right) \leq
R_X\}$\\

\subsection{Appendix A.2}
Suppose the decoding regions for $C_X$ are $D_1,D_2,...D_M$. Hence,
\begin{eqnarray}
e\left(C_X,W\right)&=&\frac{1}{M_X}\sum_{i=1}^{M_X}W\left(D_{i}^c|i\right)\nonumber\\
&=&\frac{1}{M_X}[\sum_{i=1}^{M}
\left(N_iW\left(D_i^c|\textbf{x}_i\right)+\left(N_i-1\right)W\left(D_i|\textbf{x}_i\right)\right)]\nonumber\\
&=&
\frac{1}{M_X}\left(M_X-M+\sum_{i=1}^{M}W\left(D_i^c|\textbf{x}_i\right)\right).
\end{eqnarray}
Let us randomly choose $\textbf{x} \in T_{P_X}$ that does not belong
to $C_X$. Define
\begin{eqnarray}
V_0 \triangleq
\arg\min_{V}\{D\left(V||W|P_X\right)+H\left(V|P_X\right)\}
\end{eqnarray}
it is proved that if $\textbf{y} \in T_{V_0}\left(\textbf{x}\right)$
\begin{eqnarray}
W^n\left(\textbf{y}|\textbf{x}\right)&=&2^{-n[min_{V}\{D\left(V||W|P_X\right)+H\left(V|P_X\right)\}]} \nonumber\\
&\geq& 2^{-n[\{D\left(V||W|P_X\right)+H\left(V|P_X\right)\}]} \;\;\;\;\; \text{any }  V \nonumber\\
&=& W^n\left(\textbf{y}|\textbf{x}'\right)
\end{eqnarray}
for some $\textbf{x}'$ such that $\textbf{y}\in
T_V\left(\textbf{x}'\right)$. Thus,
\begin{eqnarray}
W^n\left(\textbf{y}|\textbf{x}\right) \geq
W^n\left(\textbf{y}|\textbf{x}_i\right) \;\;\;\;\;\; \;\;\;\;\;
\;\;\;
 \text{any }  i=1,...M
\end{eqnarray}
Choose $\textbf{y} \in T_{V_0}\left(\textbf{x}\right)\cap D_k$ for
some $k$ with $|D_k| \geq 2$. Now, let us look at $\C'_X$ which
contains all codewords in $\C_X$ except one of the repeated ones,
i.e one of the $\textbf{x}_M$ which is replaced with $\textbf{x}$,
and define the decoding sets
\begin{eqnarray}
D'_i&=&D_i
\;\;\;\;\;\;\;\;\;\;\;\;\;\;\;\;\;\;\;\;\;\;\;\;\;\;\;\;\;\;\;\;\;\;\;\;\;i\neq
k\\
D'_k&=&D_k -\{\textbf{y}\}\\
D'_{M+1}&=& \{\textbf{y}\}\;\;\;\;\;\;\;\;\;\;\;\;\;\;\;\;\;\;\;\;\;
\text{where }\textbf{x}'_{M+1} \triangleq \textbf{x}.
\end{eqnarray}
By following a similar approach, we conclude that
\begin{align}
e'(C_X,W) &= \frac{1}{M_X}\Big(M_X-M+\sum_{i=1,i\neq
k}^{M}W(D_i^c|\textbf{x}_i)+W(D_k^{'c}|\textbf{x}_k)-W(\textbf{y}|\textbf{x})\Big)\nonumber\\
&=
\frac{1}{M_X}\Big(M_X-M+\sum_{i=1}^{M}W(D_i^c|\textbf{x}_i)W(D_k^{'c}|\textbf{x}_k)-W(\textbf{y}|\textbf{x})-W(D_k^c|\textbf{x}_k)\Big)\nonumber\\
&=e(C_X,W)+\frac{1}{M_X}\Big(W(\textbf{y}|\textbf{x}_k)-W(\textbf{y}|\textbf{x})\Big)\nonumber\\
&\leq e(C_X,W),
\end{align}
where the last inequality follows from the fact that
$W\left(\textbf{y}|\textbf{x}_k\right)\leq
W\left(\textbf{y}|\textbf{x}\right)$.

\subsection{Appendix A.3}
The left hand side of~\eqref{number2} is straightforward, since for
all multiuser codes, $\C$, $e_m(\C,W) \geq e(\C,W)$.
By~\eqref{number1}, for all multiuser codes with rate pair
$(R_X,R_Y)$, we can conclude that
\begin{eqnarray}
e_m(\C,W) \geq 2^{-nE^U_{m}(R_X,R_Y)}.\label{assump}
\end{eqnarray}
Let us assume that there exists a code $\C$ with rate pair
$(R_X,R_Y)$ for which the right hand side of~\eqref{number2} does
not hold. Without loss of generality, we assume $R_X \leq R_Y$. For
Therefore,
\begin{eqnarray*}
e(\C,W) < \frac{1}{2} 2^{-n \left(E^U_{m}(R_X,R_Y) + R_X \right)},
\end{eqnarray*}
which is equivalent to
\begin{eqnarray}
\frac{1}{M_XM_Y}\sum_{i=1}^{M_X}\sum_{j=1}^{M_Y}W\left(D_{i,j}^c|\mathbf{x}_i,\mathbf{x}_j\right)
< \frac{1}{2} 2^{-n \left(E^U_{m}(R_X,R_Y) + R_X \right)},
\end{eqnarray}
which can be written as
\begin{eqnarray}
\frac{1}{M_Y}\sum_{j=1}^{M_Y}\frac{1}{M_X}\sum_{i=1}^{M_X}W\left(D_{i,j}^c|\mathbf{x}_i,\mathbf{x}_j\right)
< \frac{1}{2} 2^{-n \left(E^U_{m}(R_X,R_Y) + R_X \right)},
\end{eqnarray}
therefore, there exist $M^1_Y \geq \frac{M_Y}{2}$ codewords in
$\C_Y$ that satisfy
\begin{eqnarray}
\frac{1}{M_X}\sum_{i=1}^{M_X}W\left(D_{i,j}^c|\mathbf{x}_i,\mathbf{x}_j\right)
<  2^{-n \left(E^U_{m}(R_X,R_Y) + R_X \right)}.\label{multeq}
\end{eqnarray}
Let us call this set of codewords as $C^1_Y$. By multiplying both
sides of~\eqref{multeq} with $M_X$, and considering the fact that
all terms in summation are non-negative, it can be concluded that
for every $\mathbf{x}_i \in \C_X$, $ \mathbf{y}_j \in \C^1_Y$,
\begin{eqnarray}
W\left(D_{i,j}^c|\mathbf{x}_i,\mathbf{x}_j\right) <  2^{-n
\left(E^U_{m}(R_X,R_Y) \right)}.
\end{eqnarray}
Therefore, the new multiuser code $\C^1=\C_X \times \C^1_Y$, has a
rate pair very close to the original code, and its maximal
probability of error satisfies
\begin{eqnarray}
e_m(\C^1,W) < 2^{-n \left(E^U_{m}(R_X,R_Y)
\right)}.\label{contradiction}
\end{eqnarray}
\eqref{contradiction} contradicts our assumption in~\eqref{assump},
therefore it can be concluded that the assumption must be false and
that its opposite must be true. Similarly, we can show the bounds
in~\eqref{number4} by assumption in~\eqref{number3}.

\bibliographystyle{plain}
\bibliography{ali}

\end{document}